\documentclass[sigconf]{acmart}

\usepackage{graphicx}
\usepackage{algorithm}
\usepackage{algpseudocode}
\usepackage{booktabs}
\usepackage{multirow}
\usepackage{graphicx}
\usepackage{array}

\usepackage{makecell} 
\usepackage{booktabs}
\usepackage{amsmath}
\usepackage[table]{xcolor}
\usepackage[dvipsnames]{xcolor}
\definecolor{mygray}{gray}{0.95}
\definecolor{headergray}{HTML}{C4C8D4}
\usepackage{subcaption}
\usepackage{pifont}
\usepackage{stfloats}

\usepackage{tcolorbox}
\tcbuselibrary{breakable, skins, listings}

\definecolor{Periwinkle}{HTML}{A9A9E6}
\definecolor{mygrey}{gray}{0.6}

\AtBeginDocument{%
  }

\setcopyright{cc}
\setcctype{by}
\copyrightyear{2026}
\acmYear{2026}
\acmDOI{10.1145/3770855.3817858}
\acmConference[KDD 2026] {Proceedings of the 32nd ACM SIGKDD Conference on Knowledge Discovery and Data Mining V.2}{August 9--13, 2026}{Jeju Island, Republic of Korea.}
\acmBooktitle{Proceedings of the 32nd ACM SIGKDD Conference on Knowledge Discovery and Data Mining V.2 (KDD 2026), August 9--13, 2026, Jeju Island, Republic of Korea}

\acmISBN{979-8-4007-2259-2/2026/08}

\begin{document}

\title{SelPE: Progressive Selection for Private Structured Text Synthesis}

\author{Xuancheng Zhu}
\affiliation{%
  \institution{Beijing University of Posts and Telecommunications}
  \city{Beijing}
  \country{China}
}
\email{zhuxuanch@bupt.edu.cn}

\author{Guoshun Nan}
\authornote{Guoshun Nan and Min Lei are corresponding authors.}
\affiliation{%
  \institution{Beijing University of Posts and Telecommunications}
  \city{Beijing}
  \country{China}
}
\email{nanguo2021@bupt.edu.cn}

\author{Han Zhang}
\affiliation{%
  \institution{Beijing University of Posts and Telecommunications}
  \city{Beijing}
  \country{China}
}
\email{zhangh@bupt.edu.cn}

\author{Ben Niu}
\affiliation{%
  \institution{Beijing University of Posts and Telecommunications}
  \city{Beijing}
  \country{China}
}
\email{byrne625@bupt.edu.cn}

\author{Yang Yue}
\affiliation{%
  \institution{Beijing University of Posts and Telecommunications}
  \city{Beijing}
  \country{China}
}
\email{yueyang2024@bupt.edu.cn}

\author{Zixu Wang}
\affiliation{%
  \institution{Beijing University of Posts and Telecommunications}
  \city{Beijing}
  \country{China}
}
\email{2022213257@bupt.cn}

\author{Yilian Liu}
\affiliation{%
  \institution{Beijing University of Posts and Telecommunications}
  \city{Beijing}
  \country{China}
}
\email{liuyilian@bupt.edu.cn}

\author{Min Lei}
\authornotemark[1]
\affiliation{%
  \institution{Beijing University of Posts and Telecommunications}
  \city{Beijing}
  \country{China}
}
\email{leimin@bupt.edu.cn}

\author{Xiaofeng Tao}
\affiliation{%
  \institution{Beijing University of Posts and Telecommunications}
  \city{Beijing}
  \country{China}
}
\email{taoxf@bupt.edu.cn}

\renewcommand{\shortauthors}{Xuancheng Zhu et al.}

\newcommand{\Topk}{\operatorname{Top}\text{-}k}

\begin{abstract}
Many data-driven applications rely on structured textual records, such as clinical triage notes and financial transaction logs, for downstream learning and decision-making.
In privacy-sensitive domains, access to such records is strictly regulated, often resulting in only a small number of available private examples for model development and analysis.
Yet existing differential privacy data synthesis methods fall short: tabular techniques cannot faithfully model free-form text, while text-based approaches often break structural constraints.
We propose \textbf{SelPE}, a selection-guided progressive evolution framework for small-sample private structured text synthesis.
Rather than relying on noisy aggregation or private model training, SelPE concentrates privacy budget on a sequence of multi-batch top-$1$ selections, enabling efficient guidance under tight privacy constraints.
To support faithful and valid synthesis, SelPE decouples semantic abstraction from schema realization via a two-stage generation pipeline, and evaluates candidates using a multi-channel distance kernel that jointly models textual, categorical, and numeric fields in their native representations.
A non-private contrastive expansion mechanism further promotes diversity without incurring additional privacy cost.
Extensive Experiments demonstrate that SelPE consistently improves structural validity, fidelity, and downstream utility under strict differential privacy budgets, particularly in low-data regimes. 
Our code is available at~\href{https://github.com/ZhuXuanCH/SelPE}{\textcolor{NavyBlue}{this link}}.

\end{abstract}

\ccsdesc[500]{Security and privacy~Differential privacy}
\ccsdesc[500]{Security and privacy~Privacy protections}

\keywords{Structured Text Data, Data Synthesis, Differentially Private}


\maketitle

\section{Introduction}

\begin{figure}[t]
  \centering
  \includegraphics[
    width=\linewidth,
    height=0.40\textheight,
    keepaspectratio
  ]{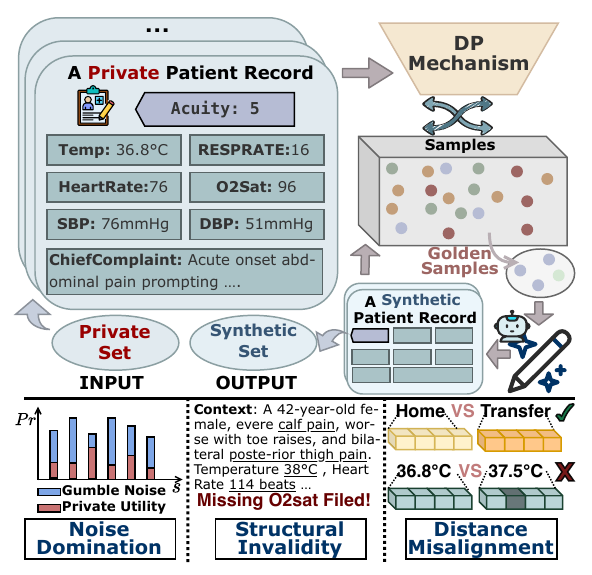}
  \caption{Overview of our differentially private synthesis pipeline. The figure illustrates the end-to-end DP synthesis pipeline and highlights three key challenges. }


  \label{fig:intro}
  \vspace{-15pt}
\end{figure}

Data has become a central asset in modern data-driven systems~\citep{souza2025breaking, fu2025ins, zhou2025anyprefer}, underpinning model training~\citep{gan2025towards, li2025forewarned, du2025graphmaster}, evaluation~\citep{ScIRGen, TimeGraph, SciHorizon}, and decision~\citep{yao2025a, xie2025xdrive, piterbarg2025training}.
In many real-world settings, data is inherently structured, consisting of multiple interdependent fields that combine numeric values, categorical attributes, and free-form text under strict schemas~\citep{xu2019modeling}.
Such structured textual data is ubiquitous and can be viewed equivalently as tables with text or structured text with rich semantics~\citep{yin-etal-2020-tabert, herzig-etal-2020-tapas}.
However, regulations and compliance requirements often restrict access to such sensitive datasets, leaving practitioners with only a small number of private examples.
This tension highlights the importance of studying small-sample structured text data synthesis, where useful and realistic data must be synthesized under tight differentially privacy (DP) constraints~\citep{pierquin2025privacy, zou2025contrastive, gao2025dataadaptive, tan2025synthesizing}.


Existing research on private data synthesis~\citep{hu2024sok} largely follows two directions: tabular data synthesis~\citep{Cormode2025tabularsurvey} and text data synthesis.
For tabular data~\citep{yang2024tabulardpsurvey}, prior methods perturb low-dimensional marginals or graphical models, or train deep generative models such as VAEs~\citep{chen2018differentially, acs2018differentially,takagi2021p3gm} and diffusion models; however, they mainly target numeric or categorical fields and do not naturally handle free-form text.
For textual data, DP-SGD~\citep{abadi2016deep} and Private Evolution (PE)~\citep{xie2024augpe} enable private generation without releasing raw data, but are primarily designed for unstructured text~\citep{gonzalez2025private}, leaving schemas and cross-field dependencies under-modeled.
Recent work~\citep{wangstruct} shows that for structured text with mixed field types and explicit schemas, existing DP synthesis methods still suffer from unstable utility estimation, semantic drift, and schema violations.
These limitations expose the lack of DP synthesis mechanisms that can effectively exploit structural information when private data are scarce.

As shown in Figure~\ref{fig:intro}, we study private structured-text synthesis under small sample sizes, which faces three key challenges.
\emph{(i) Signal sparsity.}
With limited private data, additive-noise mechanisms~\citep{geng2015optimal} and multi-output selection~\citep{qiao2021oneshot} dilute utility signals~\citep{zhangpcevolve}, allowing noise to dominate private guidance and obscure candidate distinctions.
\emph{(ii) Schema fidelity.}
Generated records must satisfy complex schema constraints~\citep{geng2025jsonschemabench} across heterogeneous fields, yet hard constraints may restrict model reasoning~\citep{banerjeecrane,gonzalez2025constrained}, whereas free-form generation often violates schemas~\citep{yang2025structeval,raspanti2025grammar}.
\emph{(iii) Cross-type alignment.}
Candidate evaluation must jointly compare free-form text, categorical attributes, and numeric fields with incompatible spaces and scales~\cite{gorishniy2021revisiting,borisov2022deep}; single-space embeddings or unified numeric metrics collapse heterogeneous signals and yield misaligned comparisons~\citep{zhu2018heterogeneous}.

A central observation is that, in the small-sample regime, noise-based aggregation mechanisms become dominated by perturbation, as the noise magnitude grows with the candidate pool size while the available signal is constrained by the limited private data (Sec.~\ref{sec:3_4_limitationsnoise}).
Moreover, text differs fundamentally from images~\citep{zhangpcevolve} or numeric data~\citep{wang2024harmonic}: textual records are diverse and compositional, and cannot be meaningfully summarized by a single average.
As a result, privacy budget is better concentrated on a small number of \emph{high-confidence decisions}.

Building on this, we propose \textbf{SelPE} (\emph{Selection-guided Progressive Evolution}), a framework for small-sample private structured text synthesis.
SelPE allocates privacy budgets to repeated \emph{multi-batch top-$1$ selections} rather than noisy aggregation.
Each private selection operates over a carefully designed mixed-type distance kernel and actively conditions future candidate generation, forming a compact evolution trajectory with high downstream impact.
First, we introduce a \emph{context–schema decoupled generation} pipeline that separates free-form semantic abstraction from schema-grounded realization, enabling diversiform semantic modeling while strictly enforcing structural validity.
Second, we propose a \emph{multi-channel distance kernel} that enables accurate comparison between structured text records by jointly modeling global semantics, field-level text, categorical attributes, and numeric values in their native representations.
Third, we develop a \emph{selection-guided progressive evolution strategy} that stabilizes selection under small samples, promotes diversity through non-private contrastive expansion, and concentrates privacy budget on the most informative decisions.

This work makes three main contributions:
\textbf{(i) Schema-aware generation and evaluation.}
We propose a cohesive framework that decouples semantic abstraction from schema-grounded realization, and supports mixed-type comparison through a multi-channel distance kernel, jointly ensuring structural validity and cross-field consistency in structured text. 
\textbf{ii) Selection-driven private evolution for structured text.}
We introduce a selection-centric PE framework tailored to small-sample structured text synthesis, replacing noise-based aggregation with multi-batch top-$1$ DP selection.
\textbf{iii) Empirical evaluation.}
We conduct extensive experiments on 3 realistic structured text datasets, comparing against 5 text data synthesis baselines under 4 privacy budgets.
Results demonstrate consistent improvements in structural fidelity and downstream utility.
We further analyze robustness across varying data scales, hyperparameter settings, and ablation variants. 
\section{Related Work}
\textbf{Private Tabular Data Synthesis.}
Prior work on DP tabular data synthesis can be broadly categorized into two directions ~\citep{Cormode2025tabularsurvey,yang2024tabulardpsurvey}.
\textbf{i)} \emph{Statistics-based approaches} privately estimate low-dimensional marginals~\citep{li2014differentially, asghar2020differentially} or graphical-model parameters~\citep{zhang2017privbayes, ma2023improved} and reconstruct synthetic datasets consistent with these statistics, offering strong privacy guarantees and competitive utility on structured tabular benchmarks.
\textbf{ii)} \emph{Deep learning approaches}~\citep{zhang2021privsyn} employ deep generative models, such as GAN, Diffusion~\citep{xie2018differentially, frigerio2019differentially} or VAE~\citep{chen2018differentially, acs2018differentially,takagi2021p3gm} trained with Differentially Private Stochastic Gradient Descent (DP-SGD)~\citep{abadi2016deep} or teacher-student frameworks, providing greater modeling flexibility but often exhibiting sensitivity to privacy budgets and training stability~\citep{frigerio2019differentially}.
Overall, existing tabular data stnthesis methods primarily target purely numeric or categorical schemas and do not address synthesis for datasets with mixed numerical and free-form textual attributes.

\textbf{Private Text Data Synthesis.}
Private text data synthesis is generally more challenging than tabular data synthesis due to the high dimensionality of language and its semantic sensitivity~\citep{hu2024sok}.
Existing work can be broadly grouped into two categories.
\textbf{i)} DP training approaches apply mechanisms such as DP-SGD ~\citep{abadi2016deep} to fine-tune language models for downstream tasks~\citep{liu2019roberta} or text generation~\citep{lilarge,yudifferentially}, providing end-to-end privacy guarantees but often exhibiting sensitivity to privacy budgets and training stability~\citep{yue2023synthetic}.
\textbf{ii)} Private Evolution (PE) methods~\citep{lin2024pe, xie2024augpe} leverage large pre-trained language models without parameter updates, iteratively selecting and refining synthetic samples via differentially private mechanisms~\citep{zoucontrastiveSynthesizing}, thereby avoiding DP training and enabling the use of strong foundation models~\citep{gonzalez2025private}.
Recent evaluations~\citep{wang2025structbench} suggest that existing private text synthesis methods primarily emphasize semantic similarity and sample-level diversity, while providing limited support for preserving structural constraints.
As a result, synthesizing structured text that follows implicit schemas~\citep{geng2025jsonschemabench,lu2025learning} or compositional rules~\citep{kairouz2015DPComposition,whitehouse2023fully} under differential privacy remains an open challenge. 
\section{Preliminaries}

\subsection{Problem Definition.}
\label{sec:3_1Definition}
\textbf{Structured Text Data.}
Building on prior formulations \citep{wangstruct}, we study \emph{structured text data} composed of multiple labeled fields with heterogeneous value types.
A schema $\mathcal{S}$ is a specification that defines a set of field labels together with field-level rules. A structured text record $s = \{(c_j, v_j)\}_{j=1}^d,$ where each value $v_j$ satisfies the constraints specified by the schema $\mathcal{S}$ and may be numeric, categorical, or free-form textual.
We denote by $\mathbb{S}$ the space of all records that conform to schema $\mathcal{S}$. When a dataset consists of $n$ such records $D = \{s_i\}_{i=1}^n \subseteq \mathbb{S}$, they can be naturally organized in a tabular form. \emph{We focus on structured text data with free-form textual fields}, as opposed to datasets consisting solely of numeric fields.

\textbf{Private Data Synthesis.}
Let $P = \{s_i\}_{i=1}^n \subset \mathbb{S} $ be a private dataset of structured text records.
The goal of differentially private data synthesis\citep{hu2024sok} is to construct a synthetic dataset
$ S = \{\tilde{s}_j\}_{j=1}^m \subset \mathbb{S} $
such that $S$ preserves fidlity, utility and privacy properties. Moreover, the algorithm is required to satisfy $(\varepsilon,\delta)$-DP with respect to the private dataset $P$.

\textbf{Access model.}
Unlike DP Training method\citep{abadi2016deep}, we consider synthesis in a privacy-preserving access regime: \emph{the algorithm does not directly reveal or train on private records, and any interaction with $P$ is mediated through DP mechanisms}. Accordingly, privacy loss is fully characterized by the sequence and composition of DP mechanisms applied to $P$.

\subsection{Differential Privacy}
\label{sec:DPtheorem}
Differential Privacy~\citep{dwork2006calibrating} provides a worst-case, distribution-independent privacy guarantee against a broad class of inference attacks~\citep{shokri2017membership,yeom2018privacy}.
Two datasets $D,D' \in \mathcal{S}^N$ are \emph{neighbors}, denoted $D \sim D'$, if they differ in exactly one record.
A randomized mechanism $\mathcal{M}$ satisfies $(\varepsilon,\delta)$-DP if for all $D \sim D'$ and all measurable events $\mathcal{E}$,
\begin{equation}
\Pr[\mathcal{M}(D)\in \mathcal{E}]
\le e^{\varepsilon}\Pr[\mathcal{M}(D')\in \mathcal{E}] + \delta .
\end{equation}

\textbf{Sequential Composition.}
DP supports adaptive sequential composition~\citep{dwork2006our}.
If a dataset is accessed by a sequence of $T$ mechanisms $\{\mathcal{M}_t\}_{t=1}^T$, each satisfying $(\varepsilon,\delta)$-DP, then their joint output satisfies $(\varepsilon_{\mathrm{tot}},\delta_{\mathrm{tot}})$-DP, where $\varepsilon_{\mathrm{tot}}$ and $\delta_{\mathrm{tot}}$ follow standard advanced composition bounds~\citep{kairouz2015DPComposition}.
Moreover, by the post-processing property~\citep{dwork2014algorithmic}, any data-independent transformation of DP outputs incurs no additional privacy loss.

\textbf{Parallel Composition.}
When DP mechanisms operate on disjoint subsets of the data~\citep{dwork2006calibrating}, 
if $D=\biguplus_{i=1}^K D_i$ and each mechanism $\mathcal{M}_i$ accesses only $D_i$ and satisfies $(\varepsilon,\delta)$-DP, then the combined mechanism
$
\mathcal{M}(D) = (\mathcal{M}_1(D_1),\dots,\mathcal{M}_K(D_K))
$
also satisfies $(\varepsilon,\delta)$-DP.
In this case, privacy loss is governed by the maximum loss of any component, rather than accumulating across mechanisms.

\textbf{Gaussian Mechanism.}
Let $f$ be a $D$-dimensional query with $\ell_2$-sensitivity
$\Delta_f = \max_{D \sim D'} \| f(D) - f(D') \|_2$.
The Gaussian mechanism releases $\mathcal{M}_{\mathrm{G}}(D) = f(D) + \mathcal{N}(0,\sigma^2 I_D)$.
For $\varepsilon \in (0,1)$ and $\delta \in (0,1)$, choosing
\begin{equation}
\label{eq:gaussian_sigma}
\sigma = \Delta_f \cdot \frac{\sqrt{2 \log(1.25/\delta)}}{\varepsilon}
\end{equation}
ensures that $\mathcal{M}_{\mathrm{G}}$ satisfies $(\varepsilon,\delta)$-differential privacy.

\textbf{Top-$1$ Exponential Mechanism.}
Let $\mathcal{R}$ be an output range and let $u:\mathcal{S}^N \times \mathcal{R} \rightarrow \mathbb{R}$ be a utility function with sensitivity
$\Delta_u = \max_{r \in \mathcal{R}} \max_{D \sim D'} |u(D,r) - u(D',r)|$.
The exponential mechanism \citep{dong2020optimal,mcsherry2007mechanism} samples an output $r \in \mathcal{R}$ according to
\begin{equation}
\label{eq:exp_mech_top1}
\Pr\!\left[\mathcal{M}_{\mathrm{E}}(D) = r\right]
=
\frac{\exp\!\left(\varepsilon\, u(D,r)/\Delta_u\right)}
{\sum_{r' \in \mathcal{R}} \exp\!\left(\varepsilon\, u(D,r')/\Delta_u\right)} ,
\end{equation}
which satisfies $\varepsilon$-differential privacy.

\subsection{Limitations of Noise-Based Mechanisms}
\label{sec:3_4_limitationsnoise}

Prior work has observed that original PE algorithm becomes ineffective in small-sample regimes, particularly for image synthesis~\citep{lin2024pe, zhangpcevolve}.
We formalize this limitation and show that noise-based DP mechanisms suffer from vanishing signal-to-noise ratio when the private sample size is small.

\begin{lemma}[Noise domination in small-sample PE]
\label{lem:small_n_limit}
Consider a private dataset of size $n$ and a candidate pool of size $M$.
For histogram-based PE mechanisms using additive Gaussian noise, the effective signal-to-noise ratio satisfies
\begin{equation}
\mathrm{SNR} = O\!\left(\frac{n}{\sigma M^{3/2}}\right).
\end{equation}
In the small-sample regime where $n \ll M$, the released utility signal is dominated by noise, leading to unreliable candidate selection.
\end{lemma}

\textbf{Proof sketch.}
Noise-based PE evaluates candidates by releasing a differentially private histogram over $M$ bins. While the perturbation magnitude scales as $\Theta(\sigma \sqrt{M})$, the total signal contributed by $n$ private records is at most $O(n)$, yielding a per-candidate signal of order $O(n/M)$.
As a result, the signal-to-noise ratio decays as $O(n/(\sigma M^{3/2}))$. A full derivation is provided in Appendix~\ref{app:noise_limit}. An analogous limitation applies to multi-output selection mechanisms based on Gumbel perturbations (top-$k$), see Appendix~\ref{app:topk}.

\section{Methodology}
\label{sec:method}

\begin{figure*}[t]
    \centering
    \setlength{\abovecaptionskip}{2pt}
    \includegraphics[width=\textwidth, trim=0 8pt 0 0, clip]{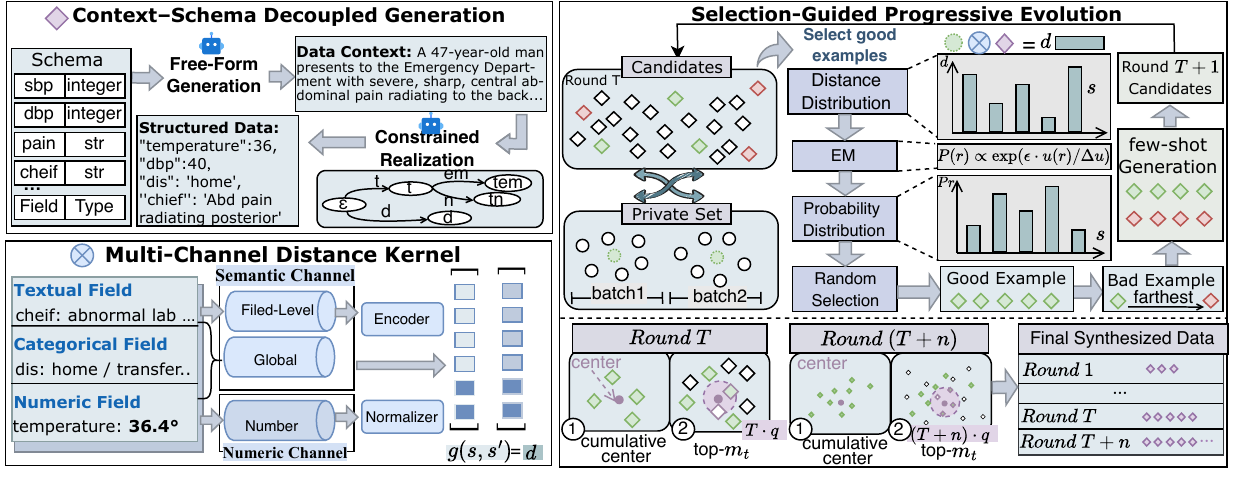}
    \caption{
    Overview of \textsc{SelPE}.
    SelPE decouples semantic abstraction from schema realization to ensure structural validity. SelPE employs a multi-channel distance kernel to jointly evaluate textual and numeric fields, and concentrates the privacy budget on a sequence of progressive selections to enable high-fidelity synthesis under tight DP constraints.
    }
    \label{fig:main}
    \vspace{-10pt}
\end{figure*}

As shown in Figure~\ref{fig:main}, SelPE is a selection-guided framework for small-sample structured text synthesis. 
Motivated by the noise domination of aggregation-based mechanisms in low-data regimes (Lemma~\ref{lem:small_n_limit}), SelPE concentrates the privacy budget on repeated multi-batch top-$1$ structure-aware selections. 
It consists of three components: 
(i) a context--schema decoupled generation pipeline that separates free-form semantic modeling from constrained realization; 
(ii) a multi-channel distance kernel that compares textual, categorical, and numeric fields in their native spaces; and 
(iii) a progressive evolution strategy that uses selected samples to guide subsequent generation while maintaining diversity.

\subsection{Context-Schema Decoupled Generation}
\label{sec:4_1two_stage}

To ensure the structural validity of LLM outputs, prior work commonly relies on constrained decoding~\citep{zhengsglang}.
However, strict schema enforcement substantially narrows the token space, often degrading semantic expressiveness and increasing decoding cost, particularly under long or complex prompts~\citep{banerjeecrane,moskal2025llguidance}.
At the same time, structured text records represent coherent real-world instances rather than independent field values~\citep{wangstruct}.
To reconcile this tension, SelPE adopts a \emph{context--schema decoupled generation} paradigm that separates semantic abstraction from schema-grounded realization.

\textbf{Stage I: Latent Data Context.}
For each synthesized sample $s$, SelPE first constructs a \emph{latent record context} $t$ that captures a coherent real-world instance and its cross-field semantics in free-form text.
Formally, let $\mathcal{M}_\theta$ be a language model with parameters $\theta$, inducing a conditional distribution $p_\theta(\cdot)$ over token sequences.
Given a prompt context, the model samples a latent description
$t \sim p_{\theta}(t \mid \text{prompt})$.
We treat $t$ as a latent semantic variable and postpone schema enforcement to the subsequent stage, thereby decoupling semantic modeling from structural realization.

\textbf{Stage II: Schema-Grounded Realization.}
In the second stage, the latent data context $t$ is realized as a structured record that strictly conforms to the target schema $\mathcal{S}$.
A schema 
$
\label{eq:schema}
\mathcal{S} = (\mathcal{C}, \{\mathcal{V}_j\}_{j=1}^d, \mathcal{R}),
$
where $\mathcal{C}=\{c_1,\dots,c_d\}$ denotes field labels, $\mathcal{V}_j$ specifies the domain of field $c_j$, and $\mathcal{R}$ encodes constraints.
Let $p_\theta(y_k \mid y_{<k}, t)$ denote the base distribution at decoding step $k$ from $\mathcal{M}_\theta$,
conditioned on the previously generated prefix $y_{<k}$ and the latent record context $t$.
Let $\mathcal{A}_{\mathcal{S}}$ be a schema-induced constraint automaton that specifies admissible next tokens given prefix $y_{<k}$.
The constrained decoding distribution is defined as
\begin{equation}
\label{eq:constrained_decoding}
p_{\theta}^{\mathcal{S}}(y_k \mid y_{<k}, t)
\;\propto\;
p_\theta(y_k \mid y_{<k}, t)\;
\mathbf{1}\!\left[y_k \in \mathcal{V}_{\mathcal{S}}(y_{<k})\right],
\end{equation}
where $\mathcal{V}_{\mathcal{S}}(y_{<k})$ denotes the set of schema-valid tokens permitted by $\mathcal{A}_{\mathcal{S}}$ after prefix $y_{<k}$, and $\mathbf{1}[\cdot]$ denotes the indicator function that masks out schema-invalid tokens.
Decoding proceeds by iteratively sampling $y_k \sim p_{\theta}^{\mathcal{S}}(\cdot \mid y_{<k}, t)$ until termination, yielding a token sequence $y_{1:L}$, which is parsed into a structured record $s \in \mathbb{S}$ according to $\mathcal{S}$. We instantiate $\mathcal{A}_{\mathcal{S}}$ using \textsc{LLGuidance}~\citep{moskal2025llguidance}.

\subsection{Multi-Channel Distance Kernel}
\label{sec:4_2distance_kernel}

Structured text consists of heterogeneous fields (textual, categorical, and numeric) that differ in representation and scale.
To enable reliable comparison under this heterogeneity, SelPE defines a \emph{multi-channel distance kernel} $g(\cdot,\cdot)$ that evaluates each field in its native space and aggregates distances via late fusion.

\textbf{Hybrid representation.}
Given a schema $\mathcal{S}$ (Eq.~\eqref{eq:schema}), each structured record
$s \in \mathbb{S}$ is mapped to
\begin{equation}
\label{eq:hybrid_repr}
\Phi(s)
=
\bigl(
\{ \mathbf{z}_{\alpha}(s) \}_{\alpha \in \mathcal{K}_z},
\;
\{ x_{\beta}(s) \}_{\beta \in \mathcal{K}_x}
\bigr),
\end{equation}
where $\mathcal{K}_z$ indexes semantic embedding channels and $\mathcal{K}_x$ indexes scalar numeric channels, explicitly separating semantic and numeric signals.

\textbf{Channel-wise encoding.}
Each field contributes one or more channels depending on its type.

\emph{Semantic channels.}
For a field $c_j$ with value $\xi_j$, SelPE constructs a semantic representation by encoding the field--value pair via a deterministic lexicalization function $\ell(c_j, \xi_j)$.
A semantic embedding is then obtained using a sentence encoder $f_{\mathrm{emb}}$ followed by $\ell_2$ normalization:
\begin{equation}
\label{eq:text_embed}
\mathbf{z}_j(s)
=
\frac{
f_{\mathrm{emb}}\!\bigl(\ell(c_j, \xi_j)\bigr)
}{
\left\|
f_{\mathrm{emb}}\!\bigl(\ell(c_j, \xi_j)\bigr)
\right\|_2
}.
\end{equation}
This construction applies uniformly to textual, categorical, and numeric fields, enabling a unified semantic treatment across heterogeneous attributes.

In addition to field-level semantic, SelPE also constructs a \emph{global semantic representation}.
Specifically, a global embedding $\mathbf{z}_{\mathrm{glob}}(s)$ is obtained by encoding a canonical linearization of the entire record, formed by concatenating $\{\ell(c_j, \xi_j)\}_{j=1}^d$ in a fixed schema order.
This global channel captures cross-field coherence beyond individual.

\emph{Numeric channels.}
For numeric fields, SelPE additionally introduces a scalar channel to preserve numerical precision,
yielding a dual-channel representation
$\bigl(\mathbf{z}_j(s),\;x_j(s)\bigr),$
where $x_j(s) \in [0,1]$ is obtained via robust normalization.
Normalization statistics are estimated over the candidate set in each round, to ensure a consistent scale during selection.

\textbf{Distance kernel.}
Given two records $s$ and $s'$, their structured distance $d(s,s')$ is computed
via the following distance kernel:
\begin{equation}
\label{eq:distance_kernel}
g(s,s')
=
\sum_{\alpha \in \mathcal{K}_z}
w_{\alpha} \,
d_{\mathrm{cos}}\!\bigl(\mathbf{z}_{\alpha}(s), \mathbf{z}_{\alpha}(s')\bigr)
+
\sum_{\beta \in \mathcal{K}_x}
w_{\beta} \,
\lvert(x_{\beta}(s) - x_{\beta}(s')\rvert,
\end{equation}
where $d_{\mathrm{cos}}(a,b)=1-\langle a,b\rangle$ and $w_{\alpha}, w_{\beta} \ge 0$ are channel weights. In practice, each distance is clipped to $[0,1]$ before aggregation, and all channel weights are set equally by default. The kernel aggregates field distances
additively, ensuring that each attribute contributes independently
to the overall similarity.

\subsection{Selection-Guided Progressive Evolution}
\label{sec:multi_batch_pe}

SelPE formulates private structured text synthesis as a \emph{selection-guided progressive evolution} process.
Instead of releasing noisy aggregates over private data, the privacy budget is concentrated on a small number of utility-based selections that directly influence subsequent generation.
This design is motivated by Lemma~\ref{lem:small_n_limit}, while long evolutionary chains yield diminishing returns in text synthesis~\citep{xie2024augpe}.
SelPE therefore favors compact yet influential evolution trajectories, where each private decision has persistent impact.

\textbf{Progressive conditional generation.}
SelPE synthesizes a dataset of total size $N$ over $T$ evolution rounds.
Rather than allocating synthesis uniformly,
we adopt an increasing-quota schedule
\begin{equation}
\label{eq:quota}
m_t = t\,q,
\qquad
\sum_{t=1}^{T} m_t = N,
\end{equation}
where $q$ is a base expansion unit.
Later rounds thus operate on a richer synthetic context formed by
previously selected high-quality samples,
yielding higher marginal utility per private selection.

At round $t$, SelPE generates a candidate pool
\begin{equation}
\label{eq:candidate_pool}
\mathcal{H}_t = \{ s_{t,1}, \dots, s_{t,K m_t} \},
\end{equation}
using the context--schema decoupled generation mechanism
(Sec.~\ref{sec:4_1two_stage}).
Candidate generation is conditioned on both
\emph{positive} (previously selected) and
\emph{negative} (contrastive) synthetic exemplars from round $t\!-\!1$,
biasing proposals toward regions of the structured space
validated by earlier private decisions while preserving exploration.
The multiplicity factor $K$ controls exploratory breadth.

\textbf{Multi-batch private selection.}
To stabilize selection under limited private data,
SelPE partitions the private dataset $P$ into $b$ disjoint batches
$\{B_i\}_{i=1}^{b}$.
Each batch is represented by a batch-wise empirical center
computed independently for each channel.
Specifically, for semantic vector channels $k \in \mathcal{K}_v$
and numeric scalar channels $\ell \in \mathcal{K}_u$,
the batch center is defined as
\begin{equation}
\label{eq:batch_center}
\Phi\!\bigl(\mu(B_i)\bigr)
=
\Bigl(
\Bigl\{\frac{\sum_{x \in B_i} v_k(x)}{\bigl\|\sum_{x \in B_i} v_k(x)\bigr\|_2}\Bigr\}_{k \in \mathcal{K}_v},
\;
\Bigl\{\frac{1}{|B_i|}\sum_{x \in B_i} u_\ell(x)\Bigr\}_{\ell \in \mathcal{K}_u}
\Bigr),
\end{equation}
corresponding to a normalized mean direction for semantic channels
and a mean aggregation for numeric channels.

Each candidate $s \in \mathcal{H}_t$ is scored against batch $B_i$ by
\begin{equation}
\label{eq:batch_utility}
u_i(s) = -\, g\!\bigl(\mu(B_i), s\bigr),
\end{equation}
where $g(\cdot,\cdot)$ is the multi-channel distance kernel.
This yields $b$ independent utility signals without aggregating
the private data into a single statistic.
Since $g(\cdot,\cdot)$ is bounded in $[0,1]$,
SelPE applies the Top-$1$ Exponential Mechanism
(Eq.~\eqref{eq:exp_mech_top1}) to each batch
with $\varepsilon_t$:
\begin{equation}
\label{eq:exp_mech_batch}
\Pr\!\left[\tilde{s}_{t,i} = s \right]
\propto
\exp\!\left(
\varepsilon_t \, u_i(s)
\right),
\qquad s \in \mathcal{H}_t.
\end{equation}
By parallel composition, the $b$ selections together satisfy
$\varepsilon_t$-DP.
The selected set
$\tilde{S}_t=\{\tilde{s}_{t,1},\dots,\tilde{s}_{t,b}\}$
thus consists of synthetic samples most aligned with the private data
under each batch view.

\textbf{Non-private contrastive evolution.}
To encourage diversity without incurring additional privacy cost,
each private winner $\tilde{s}_{t,i}$ is paired with a contrastive
synthetic sample
\begin{equation}
\label{eq:contrastive}
\hat{s}_{t,i}
=
\arg\max_{s \in \mathcal{H}_t}
g\!\bigl(s, \tilde{s}_{t,i}\bigr),
\end{equation}
which depends only on distances among synthetic candidates.
These contrastive samples push subsequent generation away from already selected modes. The paired exemplars $\{(\tilde{s}_{t,i}, \hat{s}_{t,i})\}_{i=1}^{b}$ are used as few-shot conditioning inputs for the next round.

\emph{At the end of each round, SelPE aggregates the selected candidates to construct the final synthetic dataset.}
Formally, we maintain a cumulative set of DP-selected winners
$
\mathcal{S}_{\le t} =\bigcup_{r=1}^{t} \tilde{S}_{r},
$
where $\tilde{S}_{r}$ denotes the batch-wise Top-$1$ selections obtained at evolution round $r$.
This cumulative set $\mathcal{S}_{\le t}$ is summarized by its aggregated center in the hybrid representation space.
Candidates in the current pool $\mathcal{H}_t$ are then ranked by their distance to this cumulative center, and the top-$m_t$ samples are selected for inclusion.
A full algorithmic description is provided in Appendix~\ref{app:algorithm}.
 
\section{Experiments}
\label{sec:experiment}

\subsection{Settings}
\label{sec:exp_settings}

\textbf{Datasets.}

\textbf{Datasets.}
We evaluate SelPE on three structured text datasets with strict schema constraints.
\textbf{Water}~\citep{water} contains water bottle reviews with 5 fields,
\textbf{MIMIC-ED}~\citep{mimic-iv-ed-2.2} consists of clinical triage records with 7 fields,
and \textbf{Loan}~\citep{lendingclub_kaggle} comprises LendingClub-style financial records with 13 fields.
To enrich textual content without introducing external information,
we prompt \texttt{gpt-4o-mini}~\citep{openai_gpt4o_mini} to rewrite selected fields:
the \textit{chiefcomplaint} field in MIMIC-ED is expanded into a more descriptive sentence,
and Loan records include an additional \textit{financial\_profile\_desc} field summarizing multiple credit-related attributes.

\textbf{Baselines.}
We compare SelPE with five representative private text data synthesis methods:
\textbf{DP-DS}~\citep{yudifferentially} (DP-SGD on downstream models),
\textbf{DP-Gen}~\citep{yue2023synthetic} (DP-finetuned LLM-based generation),
\textbf{Aug-PE}~\citep{xie2024augpe} (text-oriented private evolution),
\textbf{WASP}~\citep{zoucontrastiveSynthesizing} (collaborative PE with DP voting),
and \textbf{CTCL}~\citep{tansynthesizing} (DP topic modeling with fine-grained text generation).

\begin{table*}[b]
\centering
\caption{Main results on three datasets (Water, MIMIC, Loan) under different privacy budgets ($\varepsilon$).}
\label{tab:main_results}
\setlength{\tabcolsep}{4pt}
\renewcommand{\arraystretch}{1.2} 

\resizebox{\textwidth}{!}{%
\begin{tabular}{ll | ccc | ccc | ccc | ccc}
\Xhline{1.2pt} 

\rowcolor{headergray}
&
&
\multicolumn{3}{c|}{$\boldsymbol{\varepsilon=\infty}$} & 
\multicolumn{3}{c|}{$\boldsymbol{\varepsilon=4}$} & 
\multicolumn{3}{c|}{$\boldsymbol{\varepsilon=2}$} & 
\multicolumn{3}{c}{$\boldsymbol{\varepsilon=1}$} \\

\rowcolor{headergray}
\multirow{-2}{*}{\textbf{Dataset}} & 
\multirow{-2}{*}{\textbf{Method}}&
RoBERTa & TabSTAR & R-CFG & 
RoBERTa & TabSTAR & R-CFG & 
RoBERTa & TabSTAR & R-CFG & 
RoBERTa & TabSTAR & R-CFG \\

\Xhline{1.2pt} 

& DP-DS & 57.54 & 56.97 & -- & 45.81 & 52.88 & -- & 45.81 & 52.23 & -- & 45.81 & 44.61 & -- \\
\rowcolor{mygray} \cellcolor{white} 
& DP-Gen & 83.69 & \underline{61.91} & \textbf{100} & 77.65 & \underline{59.07} & \textbf{100} & 76.85 & 53.67 & \underline{99.95} & 78.03 & \underline{58.74} & \textbf{100} \\
& AUG-PE& 64.99 & 52.11 & 98.35 & 58.96 & 52.11 & 52.10 & 63.17 & \underline{54.63} & 52.15 & 64.46 & 53.93 & 53.15 \\
\rowcolor{mygray} \cellcolor{white}
& WASP& \underline{85.08} & 57.38 & 55.40 & \underline{82.15} & 57.02 & 53.25 & \textbf{85.34} & 54.15 & 62.05 & \underline{78.52} & 52.22 & 39.45 \\
& CTCL & 44.25 & 59.78 & 96.15 & 52.08 & 51.25 & 0     & 55.75 & 50.58 & 0     & 52.43 & 50.21 & 0 \\
\rowcolor{mygray} \cellcolor{white}
\multirow{-6}{*}{Water} & \textbf{Ours}& \textbf{86.95} & \textbf{66.80} & \textbf{100} & \textbf{86.62} & \textbf{65.07} & \textbf{100} & \underline{84.98} & \textbf{64.89} & \textbf{100} & \textbf{85.86} & \textbf{63.15} & \textbf{100} \\

\hline

& DP-DS & \underline{70.48} & 56.47 & -- & 55.58 & 52.91 & -- & 55.58 & 47.26 & -- & 55.58 & 48.50 & -- \\ 
\rowcolor{mygray} \cellcolor{white}
& DP-Gen & 56.22 & \underline{63.46} & 78.94 & 55.07 & \underline{53.46} & \underline{66.20} & 51.35 & 51.89 & \underline{68.05} & 50.96 & \underline{50.70} & \underline{66.20} \\
& AUG-PE & 54.22 & 56.04 & \textbf{96.95} & 54.26 & 47.53 & 54.20 & 50.06 & 51.29 & 55.10 & 40.59 & 47.21 & 56.00 \\
\rowcolor{mygray} \cellcolor{white}
& WASP   & {70.27} & 52.21 & 45.75 & \underline{63.57} & 52.55 & 49.75 & \textbf{68.02} & \underline{53.17} & 56.25 & \underline{62.85} & 43.19 & 52.75 \\
& CTCL   & 54.63 & 49.31 & 8.36  & 48.42 & 52.06 & 0     & 48.53 & 46.95 & 0     & 50.62 & 48.54 & 0 \\
\rowcolor{mygray} \cellcolor{white}
\multirow{-6}{*}{Mimic} & \textbf{Ours} & \textbf{75.15} & \textbf{68.67} & \underline{95.20} & \textbf{66.08} & \textbf{59.34} & \textbf{87.05} & \underline{62.44} & \textbf{59.65} & \textbf{91.90} & \textbf{65.06} & \textbf{53.20} & \textbf{95.80} \\
\hline

& DP-DS & 49.84 & 55.30 & -- & 47.74 & 53.57 & -- & 47.74 & \underline{52.63} & -- & 47.73 & 45.48 & -- \\
\rowcolor{mygray} \cellcolor{white}
& DP-Gen& 53.64 & \underline{50.42} & 83.08 & 54.41 & \underline{52.21} & \underline{95.37} & 54.76 & 48.86 & \underline{94.89} & 53.42 & 47.44 & \underline{94.67} \\
& AUG-PE & 49.01 & 49.47 & 90.50 & 49.33 & 49.56 & 49.54 & 48.72 & 51.46 & 49.04 & 50.01 & 49.12 & 50.50 \\
\rowcolor{mygray} \cellcolor{white}
& WASP & \textbf{58.51} & 49.64 & 14.43 & \textbf{58.22} & 51.89 & 12.11 & \textbf{58.01} & 49.64 & 22.36 & \underline{55.78} & \underline{53.11} & 18.36 \\
& CTCL & 49.86 & 48.98 & \textbf{100} & 48.64 & 52.11 & 0     & 50.67 & 46.94 & 0     & 51.00 & 49.03 & 0 \\
\rowcolor{mygray} \cellcolor{white}
\multirow{-6}{*}{Loan} & \textbf{Ours} & \underline{55.54} & \textbf{55.27} & \underline{93.21} & \underline{57.18} & \textbf{56.43} & \textbf{97.57} & \underline{55.55} & \textbf{53.21} & \textbf{98.21} & \textbf{58.27} & \textbf{53.21} & \textbf{95.00} \\

\Xhline{1.2pt} 

\end{tabular}%
}
\end{table*}


\textbf{Metrics.}
We evaluate synthetic data quality from two perspectives: \emph{utility}, which measures the preservation of task-relevant information for downstream learning, and \emph{fidelity}, which assesses the structural correctness and completeness of generated records.

\emph{Utility.}
We train downstream classifiers on synthetic data and evaluate them on held-out real test sets (see Appendix~\ref{app:dseval}).
We report ROC-AUC as the utility metric.
To capture utility under different modeling assumptions, we consider two downstream predictors.
\emph{Text-centric utility.}
We fine-tune RoBERTa~\citep{liu2019roberta} by treating each synthesized record as plain text, assessing whether semantic information is preserved for text-based models.
\emph{Structured-text utility.}
We train TabSTAR~\citep{arazi2025tabstar}, a foundation model for tabular data with textual fields, to evaluate whether structured dependencies and cross-field interactions are retained.
 For baselines that do not produce valid JSON outputs, structured fields are extracted via a regular-expression-based parser.

\emph{Fidelity.}
Beyond utility, structured text synthesis requires generated samples to satisfy schema constraints and remain directly usable.
However, many existing text-based synthesis methods do not reliably produce fully structured outputs, making strict schema validation infeasible.
We therefore report two complementary fidelity metrics based on \emph{Context-Free Grammar (CFG)} validation.
\emph{Rough-CFG (R-CFG)} measures \emph{data completeness} by checking whether required fields are present and parsable under a relaxed CFG, enabling fair comparison with partially structured baselines.
\emph{Strict-CFG (S-CFG)} follows the protocol of \citet{wangstruct} and enforces full schema validity, including structural consistency and value-range plausibility.
Since SelPE employs schema-constrained decoding, all generated samples can be evaluated under S-CFG.

\textbf{Synthesis settings.}
Expect the transfer experiments, all methods use \texttt{Llama-3.1-8B-Instruct}~\citep{dubey2024llama} as the base generator. We instantiate the semantic encoder $f_{\mathrm{emb}}$ using \texttt{gte-large-en-v1.5}.
In the main experiments, we sample \textbf{20 records per class} to form the private dataset.
All methods operate under this fixed small-sample regime.
We generate $20\times$ synthetic records per class, resulting in $N=2000$ samples for \emph{Water} and \emph{MIMIC} (5 classes) and $N=2800$ samples for \emph{Loan} (7 classes).
For SelPE, we set the candidate multiplicity $K=3$, evolution rounds $T=5$, and temperature $\tau=1.2$.
\vspace{-5pt}

\subsection{Main Results}
\label{sec:exp_main}

Table~\ref{tab:main_results} summarizes the main results across datasets whose
\emph{schema complexity} (number of fields) and \emph{text length} increase from
\emph{Water} to \emph{MIMIC-ED} and \emph{Loan}.
Across all privacy budgets, SelPE achieves the most \emph{consistent} performance
on both evaluation axes, namely RoBERTa for holistic-text utility and TabSTAR for
channel-wise structured utility, while maintaining high fidelity (R-CFG).
Notably, SelPE’s advantage becomes more pronounced as records grow longer and
more heterogeneous.
Under tight privacy budgets ($\varepsilon \in \{2,1\}$), SelPE preserves
downstream utility on \emph{MIMIC-ED} and \emph{Loan}, whereas many baselines
exhibit sharp degradation or fail to produce usable structured outputs.

We observe that some methods achieve competitive, and occasionally stronger, performance under RoBERTa-based evaluation, particularly for \emph{Water} and under non-private or mild privacy settings ($\varepsilon=\infty,4$).
This is expected, as holistic text evaluation primarily reflects global semantic fluency and benefits from increased generative diversity. However, such gains do not consistently translate to structure-aware utility. 
Under tighter privacy budgets and for more complex records, these methods degraded performance on TabSTAR, indicating difficulties in maintaining cross-field consistency.
In contrast, SelPE consistently performs well.

\textbf{Obs.~1. Increasing schema/length exposes a text–structure gap, which SelPE closes.}
As records become longer and more heterogeneous, holistic-text evaluation
(RoBERTa) can be satisfied by fluent narratives that still drift on numeric
values or cross-field relations.
TabSTAR, which consumes field-separated channels, is more sensitive to such
inconsistencies.
SelPE remains strong on TabSTAR across settings, indicating that its gains stem
from preserving structured dependencies rather than surface-level text fluency.
This gap is evident when baselines score competitively on RoBERTa but lag on
TabSTAR and fidelity.

\textbf{Obs.~2. Structural validity alone is insufficient; mixed-type alignment drives utility.}
Several baselines achieve high R-CFG in some regimes (e.g., DP-Gen on Water),
yet still underperform SelPE in downstream utility, especially on TabSTAR.
This shows that producing parseable fields does not guarantee numeric accuracy
or cross-field consistency.
SelPE’s advantage is most pronounced on structured-text utility, reflecting
better alignment across fields.

\textbf{Obs.~3. RoBERTa and TabSTAR expose complementary failure modes.}
Across datasets, some methods perform well on text-centric evaluation but
degrade on structured-text utility, revealing fluent yet inconsistent
generations.
SelPE consistently ranks at or near the top on \emph{both} RoBERTa and TabSTAR
under tight privacy, indicating that its evolution preserves both semantic
coherence and structured interactions.
Even when a baseline attains higher RoBERTa scores in intermediate regimes,
SelPE remains stronger on structure-aware utility and fidelity.

\vspace{-8pt}
\subsection{Backbone Transferability}
\label{sec:exp_backbone}

\begin{table*}[t]
\centering
\caption{Transferability of SelPE across different backbone LLMs under various privacy budgets.}
\label{tab:backbone}
\setlength{\tabcolsep}{4pt}
\renewcommand{\arraystretch}{1.2}

\resizebox{\textwidth}{!}{%
\begin{tabular}{l | ccc | ccc | ccc | ccc}
\Xhline{1.2pt}

\rowcolor{headergray}
& \multicolumn{3}{c|}{$\boldsymbol{\varepsilon=\infty}$}
& \multicolumn{3}{c|}{$\boldsymbol{\varepsilon=4}$}
& \multicolumn{3}{c|}{$\boldsymbol{\varepsilon=2}$}
& \multicolumn{3}{c}{$\boldsymbol{\varepsilon=1}$} \\

\rowcolor{headergray}
\multirow{-2}{*}{\textbf{Backbone}} 
& RoBERTa & TabSTAR & S-CFG
& RoBERTa & TabSTAR & S-CFG
& RoBERTa & TabSTAR & S-CFG
& RoBERTa & TabSTAR & S-CFG \\

\Xhline{1.2pt}

Llama-3.1-8B-Instruct        & \textbf{75.15} & \textbf{68.67} & 95.20 & \underline{66.08} & 59.34 & 87.05 & 62.44 & \underline{59.65} & 91.90 & 65.06 & 53.20 & 95.80 \\

\rowcolor{mygray} 
Gemma-2-9B-Instruct       & 64.20 & 60.21 & 49.75 & 63.57 & \underline{59.48} & 43.95 & 62.71 & 59.48 & 61.50  & 62.01 & 54.02 & 82.30  \\

Mistral-7B-Instruct & 66.95 & 66.44 & 62.35 & 64.97 & 57.15 & 39.60  & 61.79 & 53.37 & 27.20  & 57.20 & 52.36 & 93.65 \\

\rowcolor{mygray} 
Qwen-2.5-7B-Instruct         & 63.85 & 64.74 & \underline{95.65} & 64.50 & \textbf{60.10} & \textbf{99.50}  & \underline{63.20} & 59.26 & \textbf{100}   & \textbf{69.29} & \underline{57.40} & \underline{99.25} \\

Qwen-2.5-14B-Instruct        & \underline{70.24} & \underline{66.52} & \textbf{99.85} & \textbf{67.75} & 56.23 & \underline{95.60}  & \textbf{66.53} & \textbf{60.12} & \underline{99.95} & \underline{68.64} & \textbf{60.11} & \textbf{99.95} \\

\Xhline{1.2pt}
\end{tabular}%
}
\end{table*}

We study the transferability of SelPE across backbone LLMs by instantiating the two-stage generator
(Sec.~\ref{sec:4_1two_stage}) with a diverse set of instruction-tuned models:
Llama-3.1-8B-Instruct~\citep{dubey2024llama},
Gemma-2-9B-Instruct~\citep{team2024gemma},
Mistral-7B-Instruct-v0.3~\citep{jiang2023mistral7b},
Qwen-2.5-7B-Instruct, and Qwen-2.5-14B-Instruct~\citep{qwen2.5}.
All other components are kept identical. As shown in Tabel ~\ref{tab:backbone}, across a diverse set of backbone LLMs with different architectures and scales, SelPE consistently maintains strong utility and high schema fidelity under all
privacy budgets.
While stronger backbones yield higher absolute performance, the relative gains of SelPE are stable, indicating that its effectiveness is largely backbone-agnostic and stems from the proposed selection-guided evolution rather than reliance on a specific generator.

\textbf{Obs.~4. Consistent structured performance across backbones.}
As shown in Table~\ref{tab:backbone}, SelPE preserves high structured-text utility (TabSTAR) and near-perfect schema validity (S-CFG) across all tested backbones, even under tight privacy budgets ($\varepsilon=2,1$). Although RoBERTa-based utility varies with backbone capacity—reflecting differences in holistic text generation quality—the structured evaluation remains stable.
This suggests that SelPE effectively decouples structured data quality from backbone-specific generation variability, enabling reliable private synthesis under heterogeneous model choices.

\subsection{Stability under Data-Scale Variations}
\label{sec:exp_scale}

We evaluate the robustness of different methods under varying degrees of private data availability on the MIMIC dataset. Specifically, we vary the private data size \emph{per class} with $n \in \{5,10,20,35,50\}$ while fixing the synthetic size to $N=2000$ and $\varepsilon=2$.

Figure~\ref{fig:scale} analyzes how different synthesis methods respond to
increasing private sample size $n$.
Across both RoBERTa-based (holistic text) and TabSTAR-based (structured)
evaluations, SelPE consistently achieves the strongest and most stable
performance, particularly in the small- to medium-sample regimes.
While some baselines benefit from larger private datasets, their gains are
uneven across evaluation axes, whereas SelPE exhibits robust improvements that
are consistent with both semantic and structural utility.

\begin{figure}[!htbp]
  \centering
  \includegraphics[
    width=\linewidth,
    height=\textheight,
    keepaspectratio
  ]{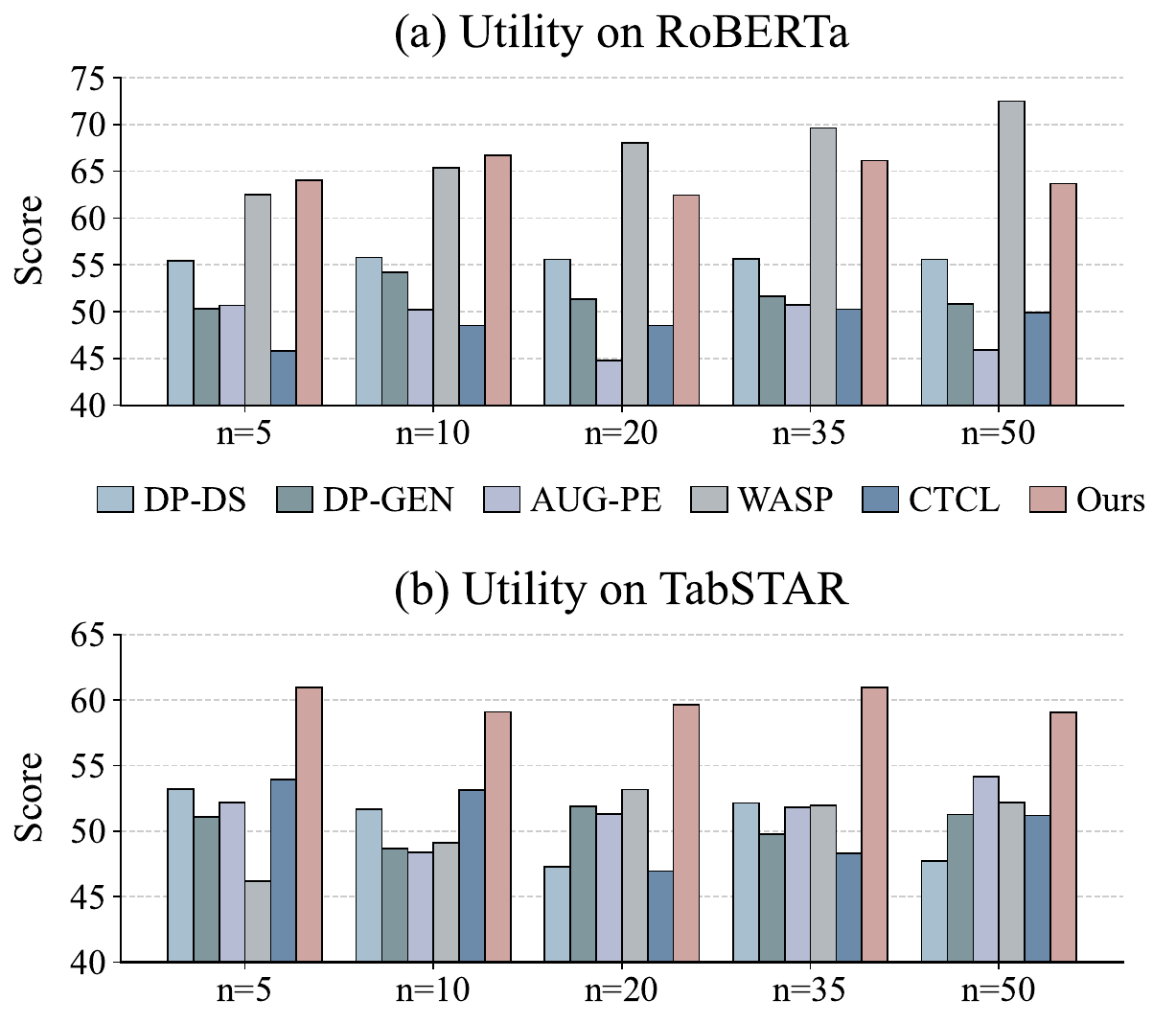}
  \caption{Downstream utility under varying data sizes $n$}
  \label{fig:scale}
  \vspace{-10pt}
\end{figure}

\textbf{Obs.~5. SelPE is robust to data scale and excels in low-$n$ regimes.}
Across $n\in[5,50]$, SelPE consistently outperforms all baselines on both
RoBERTa and TabSTAR.
Notably, SelPE already achieves strong utility at very small sample sizes
($n=5,10$), where other methods are more affected by noise and instability.
This suggests that SelPE effectively concentrates limited privacy budget on
high-impact selections, reducing sensitivity to data scale.

\textbf{Obs.~6. Model diversity mainly benefits holistic text utility.}
Methods leveraging model diversity tend to improve with larger $n$ under
RoBERTa-based evaluation, reflecting gains in global semantic richness.
However, these improvements do not consistently translate to TabSTAR,
indicating that enhanced holistic semantics alone is insufficient to ensure
robust cross-field structure.

\subsection{Hyperparameter Analysis}
\label{sec:exp_hparam}

We analyze the sensitivity of SelPE to key hyperparameters that control the evolution process, including the number of evolution rounds $T$, candidate multiplicity $K$, batch count $b$, and generation temperature $\tau$.
Figure~\ref{fig:hyper_analysis_single} reports downstream utility under each setting,
evaluated using both RoBERTa and TabSTAR classifiers.

\begin{figure}[!h]
  \centering
  
  \begin{subfigure}[b]{0.48\linewidth}
    \centering
    \includegraphics[width=\linewidth]{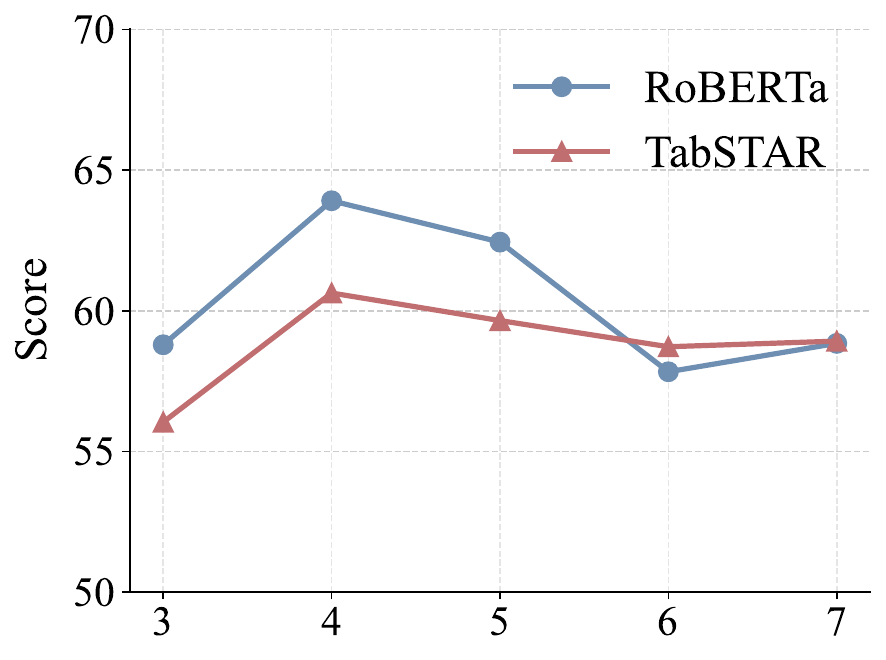}
    \caption{Evolution rounds ($T$)}
    \label{fig:hyper_round_s}
  \end{subfigure}
  \hfill 
  \begin{subfigure}[b]{0.48\linewidth}
    \centering
    \includegraphics[width=\linewidth]{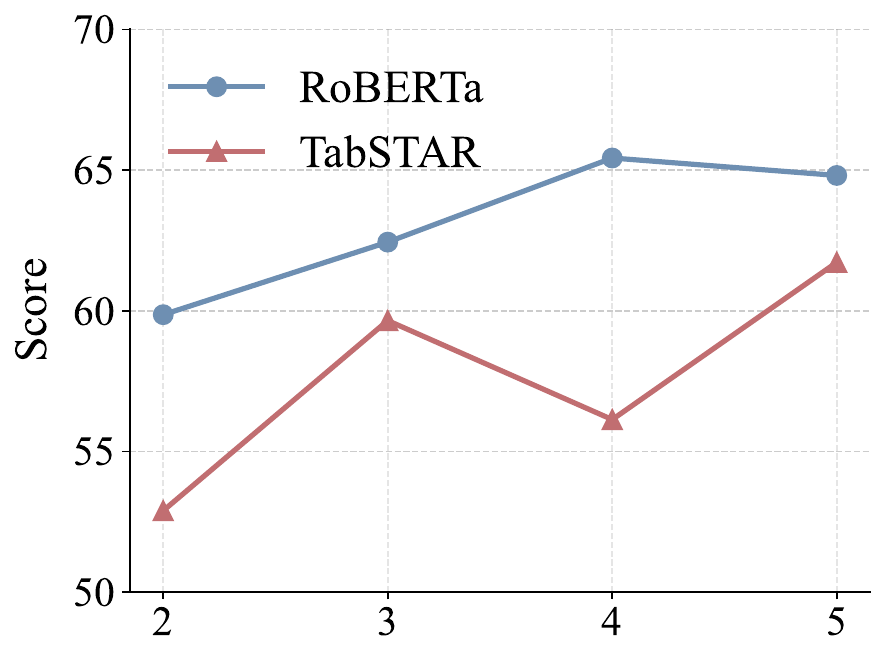}
    \caption{Candidates multiplicity ($K$)}
    \label{fig:hyper_k_s}
  \end{subfigure}
  
  \par\smallskip
  \begin{subfigure}[b]{0.48\linewidth}
    \centering
    \includegraphics[width=\linewidth]{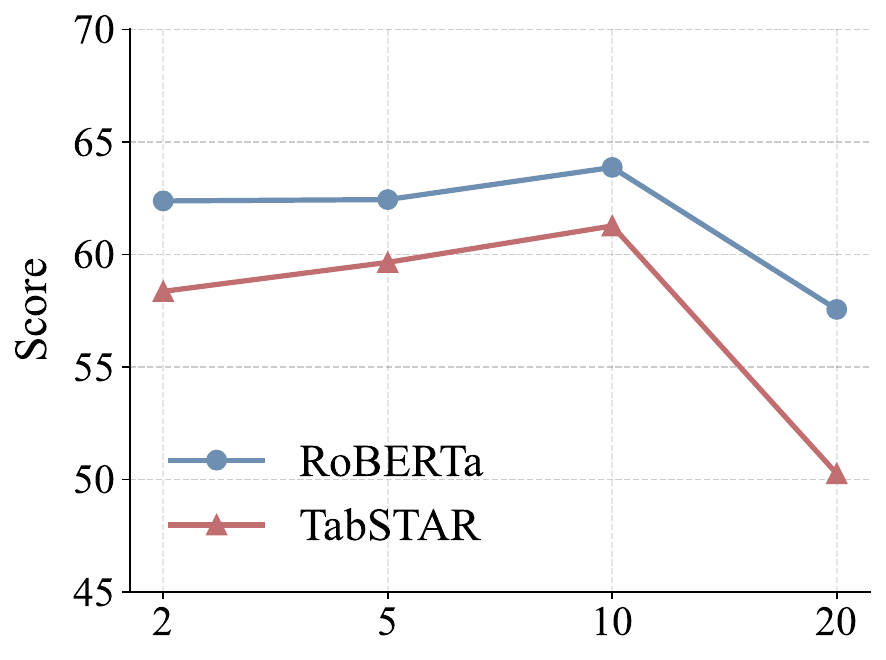}
    \caption{Batch Count ($b$)}
    \label{fig:hyper_bs_s}
  \end{subfigure}
  \hfill
  \begin{subfigure}[b]{0.48\linewidth}
    \centering
    \includegraphics[width=\linewidth]{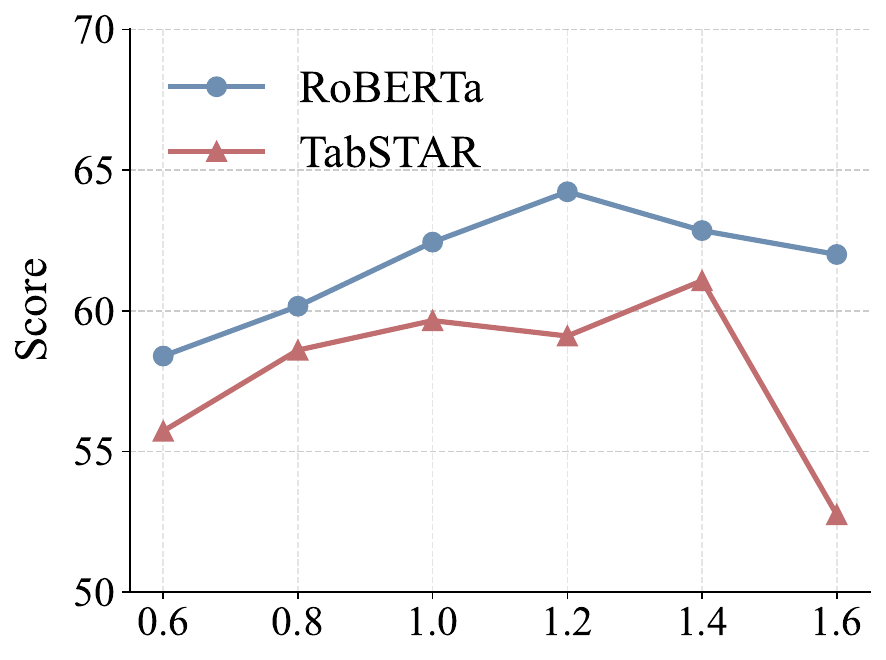}
    \caption{Temperature ($\tau$)}
    \label{fig:hyper_temperature}
  \end{subfigure}

  \caption{Hyperparameter sensitivity of SelPE.}
  \label{fig:hyper_analysis_single}
  \vspace{-10pt}
\end{figure}

\textbf{Obs.~7. Moderate evolution depth yields the best utility--stability trade-off.}
SelPE achieves peak performance with a small number of evolution rounds ($T\!\approx\!4$--$5$)
and moderate candidate multiplicity ($K\!\approx\!3$--$4$).
Increasing $T$ or $K$ beyond this range does not consistently improve utility and may even degrade performance,
suggesting diminishing returns once high-confidence selections have sufficiently shaped the generation space.
This behavior supports SelPE’s design choice of compact, selection-driven evolution
rather than deep or overly exploratory private generation.

\textbf{Obs.~8. Selection stability is sensitive to batch granularity and sampling temperature.}
Performance remains stable for small to moderate batch counts ($b\!\leq\!10$),
while overly large batches weaken selection signals by averaging heterogeneous private evidence.
Similarly, intermediate temperatures ($\tau\!\approx\!1.0$--$1.2$) provide the best balance between diversity and fidelity, whereas excessive randomness degrades structural and semantic alignment.
These trends highlight that SelPE benefits from controlled diversity,
but relies on stable, high-signal selection to guide evolution effectively. Together, these results indicate that SelPE’s performance is governed by the quality of selection signals rather than excessive exploration.

\subsection{Ablation Study}
\label{sec:exp_ablation}
We conduct ablation studies on the MIMIC dataset under $\varepsilon=2$
to assess the contribution of each key component in SelPE.
Table~\ref{tab:ablation} reports downstream utility evaluated with
RoBERTa and TabSTAR, where each variant removes exactly one design component
while keeping all other settings and privacy budgets fixed.

\begin{table}[!htbp]
    \centering
    \caption{Ablation study on different components of our method. \textbf{Bold} indicates the best performance.}
    \label{tab:ablation}
    \setlength{\tabcolsep}{8pt} 
    \renewcommand{\arraystretch}{1.2} 

    \begin{tabular}{l | c c}
        \Xhline{1.2pt}

        \rowcolor{headergray}
        \textbf{Variant} & \textbf{RoBERTa} & \textbf{TabSTAR} \\

        \Xhline{1.2pt}

        \textbf{Ours (Full)} & \textbf{62.44} & \textbf{59.65} \\

        \hline

        w/o Contrastive Example & 59.19 & 51.93 \\
        
        \rowcolor{mygray} 
        w/o Progressive Quota & 61.36 & 54.22 \\
        
        w/o Schema-aware Generation & 60.32 & 51.27 \\
        
        \rowcolor{mygray} 
        w/o Kernel $g(\cdot)$ & 62.38 & 55.13 \\

        \Xhline{1.2pt}
    \end{tabular}
\end{table}

\textbf{Obs.~9. Each component of SelPE contributes to utility, with the largest gains on structured-text evaluation.}
Removing any major component leads to a consistent performance drop, particularly on TabSTAR, indicating that SelPE’s improvements are not driven by a single design choice.
Eliminating contrastive bad-case examples causes the largest degradation, highlighting the importance of non-private contrastive signals for maintaining diversity and avoiding mode collapse.
Removing progressive quota scheduling or two-stage generation also degrades performance, suggesting that both selection accumulation and context--schema decoupling are critical for stable evolution.
Finally, removing the multi-channel distance kernel $g(\cdot)$ has limited impact on RoBERTa but substantially harms TabSTAR utility, confirming that field-aware distance modeling is essential for preserving cross-field structure rather than merely improving surface-level text quality.

\section{Conclusion}

This work introduces SelPE, a selection-guided framework for differentially private structured text synthesis tailored to the small-sample setting.
By shifting from noise-based statistical aggregation to structure-aware private selection, SelPE better preserves both semantic utility and cross-field consistency under strict privacy constraints.
SelPE combines context–schema decoupled generation with multi-channel, structure-aware selection to reliably guide the synthesis process.
Across datasets with increasing schema complexity, SelPE achieves stable performance on both holistic-text and structure-aware downstream evaluations, while maintaining high structural validity and conservative empirical privacy behavior.
Overall, our results demonstrate that concentrating the privacy budget on structured selection, together with explicitly designing structure-aware mechanisms, is an effective and principled design choice for private structured text synthesis, and represents a promising direction for DP structured text synthesis.

\begin{acks}
This work was supported by the Beijing Natural Science Foundation Program (Grant No.~L232002), the Guangxi Key Research and Development Program (Grant No.~FN2504240005), the National Natural Science Foundation of China (Grant No.~62471064), and the Fundamental Research Funds for the Beijing University of Posts and Telecommunications (Grant No.~2025AI4S02).
\end{acks}

\bibliographystyle{ACM-Reference-Format}
\bibliography{reference}

\appendix

\section{Algorithmic Details}
\label{app:algorithm}
The core algorithm of the proposed \textsc{SelPE} is described in
Sec.~\ref{sec:method}.
We include the full pseudocode here to make the procedure precise and
fully reproducible.

\begin{algorithm}[htbp]
\caption{\textsc{SelPE}: Selection-Guided Progressive Evolution}
\label{alg:sel_pe}
\begin{algorithmic}[1]
\Require $P,\ \mathcal{S},\ \mathcal{M}_\theta,\ g,\ \Phi,\ T,\ N,\ d,\ K,\ b,\ \{\varepsilon_t\}_{t=1}^T$
\Ensure Synthetic dataset $S_{\mathrm{syn}}$

\State $S_{\mathrm{syn}}\gets\emptyset$;\quad $\mathcal{E}\gets\emptyset$
\hspace{0.5em}{\footnotesize\emph{// evolution context (Sec.~\ref{sec:multi_batch_pe})}}

\For{$t=1$ \textbf{to} $T$}

    \State $m_t \gets \lceil t\cdot d \rceil$
    \hspace{0.5em}{\footnotesize\emph{// progressive quota, Eq.~\eqref{eq:quota}}}

    \State $\mathcal{H}_t \gets \textsc{Gen}(\mathcal{M}_\theta,\mathcal{S};\,\mathcal{E},\,K m_t)$
    \hspace{0.5em}{\footnotesize\emph{// conditional generation (Sec.~\ref{sec:4_1two_stage})}}

    \State Partition $P$ into $b$ disjoint batches $\{B_i\}_{i=1}^b$

    \For{$i=1$ \textbf{to} $b$}

        \State $\mu_i \gets \textsc{Proto}(B_i;\Phi)$
        \hspace{0.5em}{\footnotesize\emph{// batch prototype, Eq.~\eqref{eq:batch_center}}}

        \ForAll{$s\in\mathcal{H}_t$}
            \State $u_i(s)\gets -\,g(\mu_i,s)$
            \hspace{0.5em}{\footnotesize\emph{// utility via distance kernel, Eq.~\eqref{eq:batch_utility}}}
        \EndFor

        \State $\tilde{s}_{t,i}\sim \textsc{ExpMech}(\mathcal{H}_t,u_i,\varepsilon_t)$
        \hspace{0.5em}{\footnotesize\emph{// batch-wise top-$1$ selection, Eq.~\eqref{eq:exp_mech_batch}}}

        \State $\hat{s}_{t,i}\gets \arg\max_{s\in\mathcal{H}_t} g(s,\tilde{s}_{t,i})$
        \hspace{0.5em}{\footnotesize\emph{// non-private, Eq.~\eqref{eq:contrastive}}}

    \EndFor

    \State $\tilde{S}_t\gets\{\tilde{s}_{t,i}\}_{i=1}^b$;\quad
           $\hat{S}_t\gets\{\hat{s}_{t,i}\}_{i=1}^b$


    \Statex \hspace{1.5em}{\footnotesize\emph{// progressive inclusion and selection-conditioned evolution (Sec.~\ref{sec:multi_batch_pe})}}
    \State $S_{\mathrm{syn}}\gets S_{\mathrm{syn}}\ \cup\
       \textsc{SelectTop}(\mathcal{H}_t,\ \mathrm{Agg}(\tilde{S}_{1:t}))$
    \State $\mathcal{E}\gets \textsc{EvolveCtx}(\tilde{S}_t,\hat{S}_t,\mathcal{S})$

\EndFor

\State \Return $S_{\mathrm{syn}}$

\end{algorithmic}
\end{algorithm}

\section{Detailed Analysis of Noise-Based PE}
\subsection{Limitations of Gaussian Noise–Based PE}
\label{app:noise_limit}
\begin{proof}[Proof sketch]
Consider a private dataset $P$ of size $n$ and a synthetic candidate pool of size $M$, as commonly used in private evolution (PE).
A representative noise-based PE approach evaluates candidates by releasing a DP histogram over the candidate pool, where each bin corresponds to one synthetic candidate and records its aggregate support from the private data.

Formally, let $f(P) \in \mathbb{R}^M$ denote a histogram query over the $M$ candidates, where the $i$-th coordinate aggregates evidence from $P$ in favor of candidate $i$.
This query has $\ell_2$-sensitivity $\Delta_f$ (typically $\Delta_f = 1$).
The Gaussian mechanism releases
\begin{equation}
\tilde{f}(P) = f(P) + Z, \qquad Z \sim \mathcal{N}(0,\sigma^2 I_M),
\end{equation}
where $\sigma$ is chosen according to~\eqref{eq:gaussian_sigma}.
To characterize the perturbation magnitude, note that
\begin{equation}
\label{eq:noise_norm}
\mathbb{E}\!\left[\|Z\|_2\right]
=
\sigma \cdot \mathbb{E}\!\left[\sqrt{\chi^2_M}\right]
\approx \sigma \sqrt{M},
\end{equation}
so the typical noise energy grows with the candidate pool size $M$ and is independent of the sample size $n$.

In contrast, the signal contained in $f(P)$ is constrained by the amount of private data.
For histograms, $\|f(P)\|_1 = O(n)$; for normalized histograms, $\|f(P)\|_1 = O(1)$.
In either case, the typical per-candidate signal magnitude is at most $O(n/M)$.
As a heuristic, we define a signal-to-noise ratio (SNR) proxy as
\begin{equation}
\label{eq:snr_proxy}
\mathrm{SNR}
\;\triangleq\;
\frac{\|f(P)\|_2}{\mathbb{E}[\|Z\|_2]}
=
O\!\left(\frac{n}{\sigma M^{3/2}}\right).
\end{equation}

When $n$ is small and $M$ is large, as in small-sample private data synthesis, this ratio vanishes rapidly, and the DP-released histogram is dominated by noise. By contrast, when $n > M$, the private record dominates and histogram-based selection becomes reliable.
An analogous limitation holds for multi-output selection mechanisms based on Gumbel perturbations (top-$k$), detailed in App ~\ref{app:topk}.
\end{proof}

\subsection{Limitations of Top-k Selection via Gumbel Perturbation}
\label{app:topk}

\textbf{Top-$k$ Exponential Mechanism.}
When multiple outputs are required, top-$k$ selection under the exponential mechanism admits a one-shot implementation via independent Gumbel perturbations \citep{DPorg-one-shot-top-k}.
Specifically, sample i.i.d.\ $g_r \sim \mathrm{Gumbel}(0,\Delta_u/\varepsilon)$ for all $r \in \mathcal{R}$ and return
\begin{equation}
\label{eq:gumbel_topk}
\tilde{\mathcal{M}}^{(k)}(D)
=
\Topk\bigl(\{u(D,r)+g_r\}_{r \in \mathcal{R}}\bigr),
\end{equation}
This mechanism is distributionally equivalent to $k$ sequential exponential mechanism selections and thus admits the same privacy guarantee.

\begin{proposition}[Small-sample limitation of top-$k$ selection]
\label{prop:gumbel_topk_limit}
In the small-sample regime, top-$k$ selection via Gumbel perturbation exhibits degraded utility when applied to private evolution with a large candidate pool.
\end{proposition}
\begin{proof}[Proof sketch]
Consider the same private evolution setting as in Lemma~\ref{lem:small_n_limit}, with a private dataset $P$ of size $n$ and a candidate pool of size $M$.
Let $u(P,r)$ denote the utility score of candidate $r \in \mathcal{R}$, where $\mathcal{R}$ indexes the $M$ candidates.
Under the top-$k$ exponential mechanism, selection can be implemented by adding independent Gumbel noise
\[
g_r \sim \mathrm{Gumbel}(0,\Delta_u/\varepsilon)
\]
to each score and returning the $k$ largest perturbed values.

In private evolution, the utility differences $\Delta u_{r,r'} = u(P,r) - u(P,r')$ between competing candidates are induced by the private data and scale at most linearly with $n$.
In the small-sample regime, these differences are small and often comparable across many candidates.
By contrast, the Gumbel perturbation scale $\Delta_u/\varepsilon$ is fixed by the privacy budget and does not decrease with $n$.

As $M$ grows, the maximum and order statistics of the Gumbel noise concentrate around values on the order of $(\Delta_u/\varepsilon)\log M$.
Consequently, for sufficiently large $M$ and small $n$, the relative ordering of $u(P,r)+g_r$ is dominated by noise rather than by the data-dependent utilities.
This effect is amplified when selecting multiple outputs ($k>1$), since lower-ranked selections correspond to increasingly noise-dominated order statistics.

Therefore, in the small-sample regime, top-$k$ selection via Gumbel perturbation becomes unreliable for identifying high-utility candidates, leading to degraded selection quality.
\end{proof}

\subsection{Dataset Settings}
\label{app:dataset_settings}

We evaluate our method on three datasets: Water , MIMIC, and Loan.
The classification tasks involve 5, 5, and 7 classes, respectively.
For all datasets, we maintain a consistent split size: 500 for validation, and 1,000 for testing.
Table~\ref{tab:dataset_stats} details the class labels and specific sample distributions for the validation and test sets.

\begin{table}[htbp]
  \centering
  \caption{Detailed Dataset Statistics and Class-wise Distributions}
  \label{tab:dataset_stats}
  \resizebox{\columnwidth}{!}{
    \begin{tabular}{c c c c l}
      \toprule
      \textbf{Dataset} & \textbf{Labels} & \textbf{Split} & \textbf{Total} & \textbf{Class-wise Distribution} \\
      \midrule
      
      \multirow{2}{*}{Water} & \multirow{2}{*}{1--5} 
       & Val  & 500  & [100, 100, 100, 100, 100] \\
       & & Test & 1000 & [77, 39, 86, 222, 582] \\
      \midrule
      
      \multirow{2}{*}{MIMIC} & \multirow{2}{*}{1--5}
       & Val  & 500  & [88, 93, 94, 120, 105] \\
       & & Test & 1000 & [201, 209, 201, 195, 194] \\
      \midrule
      
      \multirow{2}{*}{Loan} & \multirow{2}{*}{1--7}
       & Val  & 500  & [69, 74, 63, 74, 81, 84, 55] \\
       & & Test & 1000 & [141, 144, 158, 125, 149, 151, 132] \\
       
      \bottomrule
    \end{tabular}
  }
\end{table}

\subsection{Downstream Evaluation Settings}

\label{app:dseval}

For multi-class classification tasks, we employ the One-vs-Rest (OvR) strategy combined with macro-averaging to compute the ROC-AUC score. Specifically, we calculate the AUC for each class against all others and reporting the arithmetic mean across all classes as follows

\begin{equation}
    \text{ROC-AUC} = \frac{1}{|\mathcal{C}|} \sum_{k \in \mathcal{C}} \text{AUC}(\mathbb{I}(y = k), \hat{p}_k),
    \label{eq:roc_auc}
\end{equation}

\noindent where $\mathcal{C}$ is the set of classes, $|\mathcal{C}|$ denotes the cardinality of $\mathcal{C}$, $\mathbb{I}(y=k)$ is the binary indicator function for class $k$ (treating class $k$ as positive and the rest as negative), and $\hat{p}_k$ represents the predicted probability for class $k$.

\section{Privacy Analysis}
\label{app:privacy}

In addition to the formal forward DP guarantees, we empirically evaluate privacy risks from memorization, near-duplicate behavior, and membership leakage.
Across datasets, SelPE achieves high Nearest Record Similarity (NRS) and competitive Distance to Closest Record (DCR), indicating negligible memorization and limited duplication.
More importantly, SelPE consistently yields lower TPR at FPR$=1\%$ under membership inference attacks than strong baselines such as AUG-PE, reflecting a more conservative and reliable privacy--utility trade-off under stringent false-positive constraints.

\textbf{Nearest Record Similarity (NRS).}
NRS measures exact memorization by checking whether synthetic records replicate training instances after key-wise normalization.
Let $n_{\mathrm{syn}}$ be the number of synthetic samples and $n_{\mathrm{dup}}$ the number of exact duplicates.
We define
$
\mathrm{NRS} = 1 - \frac{n_{\mathrm{dup}}}{n_{\mathrm{syn}}}.
$
Higher NRS indicates fewer exact replicas of the training data.

\textbf{Distance to Closest Record (DCR).}
DCR measures near-duplicate behavior by computing the distance from each synthetic record to its nearest training record.
Each record is represented by its full text and encoded with TF--IDF.
Given the training set $\mathcal{D}_{\mathrm{tr}}=\{x^{(j)}_{\mathrm{tr}}\}_{j=1}^{N}$ and the synthetic set $\mathcal{D}_{\mathrm{syn}}=\{x^{(i)}_{\mathrm{syn}}\}_{i=1}^{M}$, DCR is defined as
\[
\mathrm{DCR}\!\left(x^{(i)}_{\mathrm{syn}}\right)
=
\min_{x^{(j)}_{\mathrm{tr}}\in\mathcal{D}_{\mathrm{tr}}}
\left\lVert
\phi\!\left(x^{(i)}_{\mathrm{syn}}\right)
-
\phi\!\left(x^{(j)}_{\mathrm{tr}}\right)
\right\rVert_{1},
\]
where $\phi(\cdot)$ denotes TF--IDF vectorization.
Larger DCR suggests lower near-duplicate risk.

\textbf{Membership Inference Attack TPR at FPR$=1\%$.}
We evaluate membership leakage using a distance-based membership inference attack.
For each queried record, we use its distance to the closest synthetic record as the attack score and report the true positive rate at a fixed false positive rate of $1\%$.
Lower TPR indicates stronger resistance to membership inference.
Non-members are sampled from held-out records of the same data distribution, with the same size as the member set, to avoid bias from dataset size or distribution.

\begin{table}[h]
\centering
\caption{Privacy evaluation across datasets.}
\label{tab:privacy_all}
\setlength{\tabcolsep}{4.5pt}
\renewcommand{\arraystretch}{1.05}
\small
\begin{tabular}{llccc}
\Xhline{1.1pt}
\rowcolor{headergray}
\textbf{Dataset} & \textbf{Method} & \textbf{NRS$\uparrow$} & \textbf{DCR$\uparrow$} & \textbf{TPR$\downarrow$} \\
\Xhline{1.1pt}
\multirow{5}{*}{Water}
& DP-GEN & 1.00 & 8.16  & 3.00  \\
& AUG-PE & 1.00 & 9.05  & 8.00  \\
& WASP   & 0.99 & 8.71  & 14.00 \\
& CTCL   & 1.00 & 7.04  & 7.00  \\
& Ours   & 1.00 & 7.47  & 8.00  \\
\hline
\multirow{5}{*}{MIMIC}
& DP-GEN & 1.00 & 10.13 & 12.00 \\
& AUG-PE & 1.00 & 10.03 & 36.00 \\
& WASP   & 0.99 & 9.90  & 16.00 \\
& CTCL   & 1.00 & 7.92  & 0.00  \\
& Ours   & 1.00 & 9.99  & 10.00 \\
\hline
\multirow{5}{*}{Loan}
& DP-GEN & 1.00 & 16.02 & 12.86 \\
& AUG-PE & 1.00 & 15.51 & 20.00 \\
& WASP   & 0.99 & 16.02 & 41.43 \\
& CTCL   & 1.00 & 14.04 & 2.86  \\
& Ours   & 1.00 & 16.20 & 11.43 \\
\Xhline{1.1pt}
\end{tabular}
\end{table}

Table ~\ref{tab:privacy_all} reports privacy evaluation results on all three datasets. Across datasets, SelPE consistently achieves NRS values close to 1 and competitive DCR, indicating the absence of exact memorization and limited near-duplicate behavior.
More importantly, SelPE yields substantially lower membership inference risk under strict low-FPR attacks, as reflected by lower TPR at FPR$=1\%$ compared to strong baselines.
Overall, these results confirm that SelPE provides robust and consistent empirical privacy protection with a more conservative privacy--utility trade-off.

\section{End-to-End Privacy Guarantee}
\label{app:end_to_end_privacy}

In \textsc{SelPE}, the only operations that access the private dataset $P$ are the round-wise multi-batch \textbf{top-1 Exponential Mechanism} selections. At round $t$, the private data are randomly partitioned into $b$ disjoint batches, written as $P=\bigsqcup_{i=1}^b B_{t,i}$. Each batch-wise mechanism is then applied to its own subset:
$M_{t,i}(B_{t,i})=\mathrm{EM}(H_t,u_{t,i},\varepsilon_t)$, where $H_t$ is the candidate pool and $u_{t,i}(s)=-g(\mu(B_{t,i}),s)$ is the utility defined from batch $B_{t,i}$. Since these $b$ mechanisms operate on \textbf{disjoint subsets}, by \textbf{Parallel Composition} their joint output for round $t$, namely
$M_t(P)=(M_{t,1}(B_{t,1}),\ldots,M_{t,b}(B_{t,b}))$, is still
$\varepsilon_t$-DP. Across the $T$ rounds, the private data are accessed
repeatedly, so by \textbf{Sequential Composition} the full pipeline
$M_{1:T}(P)=(M_1(P),\ldots,M_T(P))$ satisfies total privacy
$\varepsilon_{\mathrm{tot}}=\sum_{t=1}^T \varepsilon_t$. Under the current
design, we use uniform allocation, i.e.,
$\varepsilon_t=\varepsilon_{\mathrm{tot}}/T$. All later steps, including cumulative winner aggregation, final inclusion/ranking, and contrastive
expansion, depend only on these privatized outputs, i.e.,
$S_{\mathrm{final}}=F(M_{1:T}(P))$ for some deterministic or randomized
post-processing map $F$. Therefore, by the \textbf{post-processing} property of DP, these steps incur \textbf{no additional privacy cost}.

\section{Additional Implementation Details}
\label{app:github_details}

Due to the page limit, additional implementation details are provided in this repository:
\textcolor{NavyBlue}{\url{https://github.com/ZhuXuanCH/SelPE}}.

\end{document}